\documentclass[11pt]{article}
\usepackage[left=1in,right=1in,top=0.9in,bottom=0.9in]{geometry}
\usepackage{amsmath,amsthm,amssymb}
\usepackage[colorlinks=true,citecolor=blue,linkcolor=blue,urlcolor=blue]{hyperref}

\newtheorem{theorem}{Theorem}
\newtheorem{lemma}[theorem]{Lemma}
\newtheorem{corollary}[theorem]{Corollary}
\newcommand{\mms}{\operatorname{MMS}}

\title{1-out-of-5 Maximin-Share Allocations Always Exist for Four Agents}
\author{Christoph Schwerdtfeger\\
  \texttt{christoph@kontingenz.com}}
\date{2026-07-16}

\begin{document}

\maketitle

\begin{abstract}
For four agents with nonnegative additive valuations, a complete 1-out-of-5
maximin-share allocation always exists, improving the previous 1-out-of-6
guarantee.  Together with known exact-MMS counterexamples, this completely
characterizes the four-agent case: the guarantee holds exactly for $d\geq5$.
The main technical contribution is a balanced-residual partition lemma:
removing rejected bundles with one of the four highest-ranked goods apiece
leaves a remainder that still admits the required number of unit-valued
balanced bundles.  In its central $2+2$ case, three unit bundles repair two
pairs of colliding high-valued goods.  The theorem is machine-checked in Lean~4.
\end{abstract}

\section{Introduction}

The maximin share formalizes a simple precaution.  An agent partitions the
goods into one bundle per agent, expecting to receive the bundle she values
least.  Her maximin share is the greatest value she can secure in this way.
This is an appealing fairness benchmark, but an allocation that gives every
agent her maximin share need not exist~\cite{kurokawa2018fair}.

There are two standard relaxations.  A multiplicative approximation promises
a fraction of the original share.  An \emph{ordinal} approximation instead
lets the agent partition the goods into more bundles: her 1-out-of-$d$
maximin share is the best least-bundle value in a $d$-partition, although the
final allocation still has only $n$ recipients.  Larger $d$ weakens the
benchmark without replacing it by a fixed cardinal fraction.

Budish proposed the 1-out-of-$(n+1)$ benchmark in the context of approximate
competitive equilibrium from equal incomes~\cite{budish2011combinatorial}.
Hosseini, Searns, and Segal-Halevi developed the general ordinal-MMS framework
and proved that a 1-out-of-$\lfloor 3n/2\rfloor$ allocation always
exists~\cite{hosseini2022ordinal}.  Akrami, Garg, Sharma, and Taki later gave
the asymptotically stronger general bound
$4\lceil n/3\rceil$~\cite{akrami2024improving}.  At $n=4$ this formula gives
denominator $8$, so the Hosseini--Searns--Segal-Halevi denominator $6$ remained
best.  Recent work continued to list universal 1-out-of-$(n+1)$ existence as
unresolved~\cite{akrami2026simultaneous}.

For four agents, only one denominator remained unsettled: exact MMS, or
1-out-of-4 MMS, can fail, whereas 1-out-of-6 MMS is always attainable.  We
close the gap by proving that every finite four-agent instance admits a
complete 1-out-of-5 maximin-share allocation.  Since the benchmark weakens
with the denominator, five is the smallest universal denominator.

The proof uses the normalization, ordering, and restricted Lone Divider
architecture of Hosseini, Searns, and Segal-Halevi.  In Lone Divider, one
remaining agent partitions the residual goods into acceptable bundles; an
envy-free matching allocates a nonempty subfamily and the procedure repeats.
For indivisible goods, the threshold is useful only if every future divider
can repartition what remains after previously allocated bundles are removed.
Plain Lone Divider fails this test at denominator $2n-3$, exactly $5$ when
$n=4$~\cite{hosseini2022ordinal}.  Following Hosseini, Searns, and
Segal-Halevi, we require every allocated bundle to contain exactly one of the
four highest-ranked goods.  We call these goods \emph{anchors} and such
bundles \emph{balanced}.  Removing $k$ balanced bundles removes exactly $k$
anchors, leaving one anchor for each of the $4-k$ residual bundles that the
next divider must form.

The main lemma says that balance is enough.  Normalize a divider's value so
that five witness cells each have value one.  If $k$ balanced bundles have
already been allocated and the divider values each below one, then the
residual goods can be partitioned into $4-k$ balanced bundles that she values
at least one.  For $k\geq1$, minimal filling and value accounting handle the
residual problem.  The hard case is the initial four-way split, where two
witness cells may each contain two anchors.  Separating the smaller anchor
from each collision creates four complementary deficits.  Three anchor-free
unit cells meet them: two nearest-crossing cuts leave controlled errors, and
the third cell repairs both at once.

Once this single-agent lemma is available, the allocation argument is brief.
An envy-free matching preserves every remaining agent's rejection of earlier
bundles, so the lemma can be invoked again.  A maximal partial allocation
cannot leave a positive-share agent unmatched.  Zero shares are handled in
the same matching argument and require no reduction to a three-agent theorem.

\section{Model}

Let $N=\{1,2,3,4\}$ be the agents and let $M$ be a finite set of indivisible
goods.  Agent $i$ has a nonnegative additive valuation $v_i$; thus

\[
  v_i(S)=\sum_{g\in S}v_i(g)
  \qquad\text{for every }S\subseteq M.
\]

Write $\Pi_d(X)$ for the labeled $d$-partitions of a set $X$.  Empty cells
are allowed.  The 1-out-of-$d$ maximin share of agent $i$ from $X$ is

\[
  \mms_i^d(X)
  =\max_{(P_1,\ldots,P_d)\in\Pi_d(X)}
       \min_{j\in\{1,\ldots,d\}} v_i(P_j).
\]

All maxima are attained because $M$ is finite.  An allocation is a labeled
partition $(A_i)_{i\in N}$ of $M$; empty bundles are allowed and no good is
discarded.  It is a 1-out-of-$d$ MMS allocation if

\[
  v_i(A_i)\geq \mms_i^d(M)
  \qquad\text{for every }i\in N.
\]

Our result is the following.

\begin{theorem}[Four-agent 1-out-of-5 MMS]\label{thm:main}
For every finite set $M$ and every profile of nonnegative additive valuations
$(v_1,v_2,v_3,v_4)$, there is a partition
$(A_1,A_2,A_3,A_4)$ of $M$ such that
\[
  v_i(A_i)\geq \mms_i^5(M)
  \qquad\text{for every }i\in N.
\]
\end{theorem}

\section{Proof}

The proof has two layers.  First we change the representation of the instance:
each positive share is normalized to one, and every agent's goods are placed
on a common rank scale.  We then solve the ordered normalized instance by
balanced residual partitions and envy-free matching.

\subsection{Normalization and ordering}

Set $t_i=\mms_i^5(M)$.  Dividing every value by $t_i$ would leave the total
normalized value uncontrolled, since witness cells may contain surplus.  We
instead rescale each witness cell to value exactly one.  This only lowers
values relative to $v_i/t_i$, so any bundle of normalized value at least one
still meets the original target.

\begin{lemma}[Dominated unit-witness normalization]\label{lem:normalization}
If $t_i>0$, there are a nonnegative additive valuation $w_i$ on $M$ and a
five-partition $(P_{i1},\ldots,P_{i5})$ such that
\[
  w_i(P_{ij})=1\quad(j=1,\ldots,5),
  \qquad w_i(M)=5,
  \qquad \mms_i^5(M;w_i)=1,
\]
and, for every $S\subseteq M$,
\begin{equation}\label{eq:domination}
  t_i w_i(S)\leq v_i(S).
\end{equation}
Here $\mms_i^5(M;w_i)$ denotes the maximin share computed using $w_i$.
\end{lemma}

\begin{proof}
Choose a five-partition $(P_{i1},\ldots,P_{i5})$ attaining $t_i$, and define

\[
  w_i(g)=\frac{v_i(g)}{v_i(P_{ij})}
  \qquad\text{when }g\in P_{ij}.
\]

Every denominator is positive.  Each witness cell has $w_i$-value exactly
one, and therefore $w_i(M)=5$.  The displayed partition proves that the
1-out-of-5 MMS under $w_i$ is at least one; averaging over the total value
five proves the reverse inequality.  Finally, $v_i(P_{ij})\geq t_i$, so
$t_iw_i(g)\leq v_i(g)$ item by item.  Additivity gives
\eqref{eq:domination} for every bundle.
\end{proof}

For every $i$ with $t_i>0$, fix $w_i$ and a unit witness supplied by
Lemma~\ref{lem:normalization}.  If $t_i=0$, set $w_i\equiv0$.  A bundle of
$w_i$-value at least one is therefore worth at least $t_i$ under $v_i$.

If every $t_i$ is zero, any labeled four-partition of $M$ proves the theorem.
We may thus assume that some target is positive.  Its five unit witness cells
are nonempty, so $m=|M|\geq5$.

We next order the instance.  For each agent separately, list the multiset of
values of the original goods under $w_i$ in nonincreasing order,

\[
  a_{i1}\geq a_{i2}\geq\cdots\geq a_{im},
  \qquad m=|M|,
\]

and define an ordered valuation $\bar w_i$ on common rank goods
$g_1,\ldots,g_m$ by $\bar w_i(g_r)=a_{ir}$.  All agents now agree that $g_1$
is the highest rank, $g_2$ the next highest, and so on, even though their
cardinal values remain different.  Ordering only permutes the multiset of
$w_i$-item values.  Hence, for every positive-target agent, the images of her
five witness cells under the corresponding permutation form a five-partition
of the rank goods into $\bar w_i$-unit cells.

Any allocation of the rank goods transfers back without loss.  Process ranks
$g_1,g_2,\ldots,g_m$ in this order.  When rank $g_r$ belongs to agent $i$,
let $i$ choose a remaining original good that maximizes $w_i$.  At most $r-1$
goods have already been taken, so at least one of her top $r$ original goods
is still available.  Her choice therefore has $w_i$-value at least
$a_{ir}=\bar w_i(g_r)$.  If $A_i$ is her bundle of rank goods and $S_i$ the
resulting bundle of original goods, summing over her turns gives
\begin{equation}\label{eq:ordered-transfer}
  w_i(S_i)\geq \bar w_i(A_i).
\end{equation}
This is the standard ordered-instance reduction~\cite{hosseini2022ordinal}.

It remains to solve the ordered instance.  For every positive-target agent,
the ordered valuation $\bar w_i$ has total value five and admits a
five-partition into unit cells.  In particular every item has value at most
one.  We call $H=\{g_1,g_2,g_3,g_4\}$ the set of \emph{anchors}.  A bundle is
\emph{balanced} if it contains exactly one anchor.  The set $H$ is common to
all agents; only its values differ.

\subsection{Balanced residual partitions}

Fix one positive-target agent and suppress her index.  We write $v$ for her
ordered normalized valuation.  Thus $v(M)=5$, every item is worth at most
one, and there is a witness partition into five cells of value exactly one.

We begin with the inequality that controls the waste in the filling
arguments.

\begin{lemma}[Two nonanchors]\label{lem:tail-pair}
If $x$ and $y$ are distinct goods outside $H$, then
\[
  v(x)+v(y)\leq1.
\]
\end{lemma}

\begin{proof}
Consider the six goods consisting of the four anchors together with $x$ and
$y$.  Two of them lie in the same cell of the five-cell unit witness.  Their
combined value is at most one.  On the other hand, every anchor is at least as
valuable as both $x$ and $y$.  Consequently every pair among these six goods
has combined value at least $v(x)+v(y)$: if, say, $v(x)\leq v(y)$, then all
five goods other than $x$ have value at least $v(y)$.  The colliding pair
therefore proves the claim.
\end{proof}

We shall repeatedly use the following minimal filling rule.  Start with an
anchor $a$ and a set $F$ of available nonanchors such that
$v(a)+v(F)\geq1$.  Choose an inclusion-minimal $S\subseteq F$ for which
$v(a)+v(S)\geq1$.  If $v(a)=1$, take $S=\varnothing$ and assign the fill a
\emph{virtual trigger} of value $0$.  Otherwise choose any $x\in S$ as its
trigger.  Minimality gives

\begin{equation}\label{eq:fill-bound}
  1\leq v(\{a\}\cup S)<1+v(x).
\end{equation}

The trigger value bounds the fill's overshoot above one.  Real triggers are
nonanchors, and successive ones are distinct because each filled bundle is
removed before the next fill.

The final repair rests on the following splitting principle.

\begin{lemma}[Nearest-crossing split]\label{lem:nearest-crossing}
Let $B$ be a bundle with $v(B)=1$, and suppose every good in $B$ has value at
most $\mu$.  For every $t$ with $0<t<1$, there are a bipartition $(L,R)$ of
$B$ and an error $0\leq\varepsilon\leq\mu/2$ such that either
\[
  v(L)\geq t
  \quad\text{and}\quad
  v(R)\geq1-t-\varepsilon,
\]
or
\[
  v(L)\geq t-\varepsilon
  \quad\text{and}\quad
  v(R)\geq1-t.
\]
\end{lemma}

\begin{proof}
Order the goods of $B$ arbitrarily and stop when a prefix first reaches $t$.
Let $p<t$ be the value before the crossing good and let $h$ be that good's
value, so $p+h\geq t$.  Cutting before the crossing good leaves the first
part short by $t-p$ and makes the complement meet its demand $1-t$.  Cutting
after it makes the first part meet $t$ and leaves the complement short by
$p+h-t$.  The two possible errors sum to $h\leq\mu$; choosing the smaller
gives an error at most $\mu/2$.
\end{proof}

\begin{lemma}[Balanced residual lemma]\label{lem:residual}
Let $D_1,\ldots,D_k$, where $0\leq k<4$, be pairwise disjoint balanced
bundles with $v(D_j)<1$ for every $j$.  The residual goods
\[
  R=M\setminus\bigcup_{j=1}^kD_j
\]
can be partitioned into $4-k$ balanced bundles, each of value at least one.
\end{lemma}

\begin{proof}
The removed bundles contain $k$ distinct anchors, so $R$ contains exactly
$4-k$ anchors.  We first dispose of the cases $k\geq1$.

If $k=3$, then $v(R)>5-3=2$ and $R$ contains one anchor.  The entire residual
set is the required bundle.

Suppose $k=2$.  Then $v(R)>3$ and two anchors remain.  One anchor together
with all residual nonanchors has value at least one: otherwise this set and
the other anchor, whose value is at most one, would give $v(R)<2$.  Fill the
first anchor minimally and let $\tau$ be its trigger value, with $\tau=0$
for a virtual trigger.  The filled bundle has value at most $1+\tau$, and
$\tau\leq1$.  The bundle left for the second anchor therefore has value
strictly greater than
\[
  3-(1+\tau)=2-\tau\geq1.
\]

Suppose $k=1$.  Now $v(R)>4$ and three anchors remain.  The first anchor can
be filled, since otherwise its candidate set together with the other two
anchors would have value below three.  After a first fill with trigger
$\tau_1$, the residual value is greater than $3-\tau_1\geq2$; hence a second
anchor can also be filled.  Let its trigger be $\tau_2$.  The final bundle
has value greater than
\[
  4-(1+\tau_1)-(1+\tau_2)
  =2-(\tau_1+\tau_2).
\]
If both triggers are real, they are distinct nonanchors, so
$\tau_1+\tau_2\leq1$ by Lemma~\ref{lem:tail-pair}.  If exactly one trigger is
real, its value is at most one because every item belongs to a unit witness
cell; if both are virtual, their sum is zero.  Thus in every case
$\tau_1+\tau_2\leq1$, and the final bundle has value at least one.

It remains to treat $k=0$.  Write
\[
  m_0=v(g_4),
\]
the value of the least valuable anchor.  Every nonanchor has value at most
$m_0$.  The cutoff $1/3$ is dictated by the accounting from both directions:
below it, three minimal fills can overshoot by at most one in total; above it,
no unit witness cell can contain three anchors.

First suppose $m_0\leq1/3$.  Fill three anchors in sequence and give all
remaining goods to the fourth.  None of the three fills can stall.  Before
the $j$th fill, where $j\in\{1,2,3\}$, the previous bundles have total value
at most $(j-1)(1+m_0)$.  If the current anchor and all available nonanchors
had value below one, then the remaining $4-j$ anchors, each worth at most
one, would give
\[
  v(M)<(j-1)(1+m_0)+1+(4-j).
\]
For $j=1,2,3$ the right-hand side is respectively
$4$, $4+m_0$, and $4+2m_0$, always below $5$.  This contradicts $v(M)=5$.
Each trigger is at most $m_0$, so the last bundle has value at least
\[
  5-3(1+m_0)=2-3m_0\geq1.
\]

Now suppose $m_0>1/3$.  Inspect the locations of the four anchors in the
five unit cells of the witness partition.  No cell contains three anchors.
Ignoring cells containing no anchors, the possible occupancy patterns are
therefore
\[
  1+1+1+1,\qquad 2+1+1,\qquad 2+2.
\]

For the pattern $1+1+1+1$, take the four anchored witness cells and merge the
anchor-free cell into any one of them.

For the pattern $2+1+1$, keep the two singleton-anchor cells.  Let $C$ be the
collision cell, with anchors $a,b$, and let $U,V$ be the two anchor-free unit
cells.  The other two balanced bundles are
\[
  U\cup\{a\}
  \qquad\text{and}\qquad
  V\cup(C\setminus\{a\}).
\]
Their values are at least $1$ and $2-v(a)\geq1$, respectively, and together
with the singleton cells they cover all goods.

The pattern $2+2$ is the only case in which unit cells must be split.  Let
$C,C'$ be the two collision cells and $U,V,W$ the three anchor-free unit
cells.  In $C$, let $a$ be the less valuable anchor and set $x=v(a)$; define
$b$ and $z=v(b)$ similarly in $C'$.  Since $a$ and $b$ are anchors but are the
smaller members of pairs whose total value is at most one,
\begin{equation}\label{eq:xz-range}
  m_0\leq x,z\leq\frac12.
\end{equation}
Regard the collision cells as four one-anchor cores:
$C\setminus\{a\}$, $\{a\}$, $C'\setminus\{b\}$, and $\{b\}$.  Their values are
\[
  1-x,\quad x,\quad 1-z,\quad z.
\]
To raise all four cores to one, we need pieces of respective values
\[
  x,\quad 1-x,\quad z,\quad 1-z
\]
from $U\cup V\cup W$.

Use $U$ for the two complementary demands $x,1-x$ arising from $C$, and use
$V$ for the demands $z,1-z$ arising from $C'$.  Applying
Lemma~\ref{lem:nearest-crossing} with $\mu=m_0$ partitions each unit cell into
two pieces: one meets its assigned demand, while the other may be short.  Pair
the two pieces from $U$ with the two cores from $C$ in the orientation supplied
by the lemma, and similarly pair the pieces from $V$ with the cores from $C'$.
Let the two possible shortfalls be $e$ and $f$.  Then
\[
  0\leq e,f\leq m_0/2,
  \qquad e+f\leq m_0.
\]
Order the goods of $W$ arbitrarily and add them to $Q$ until its value first
reaches $e$, taking $Q=\varnothing$ if $e=0$.  The last item added is worth
at most $m_0$, so
\[
  e\leq v(Q)\leq e+m_0.
\]
The complement has enough value to repair the other shortfall, because
\[
  v(W\setminus Q)
  \geq1-e-m_0
  \geq f;
\]
the final inequality is equivalent to
$e+m_0+f\leq1$, which follows from
$e+f\leq m_0$ and $m_0\leq1/2$.

Attach $Q$ and $W\setminus Q$ to the two possibly deficient pieces.  We have now
partitioned $U\cup V\cup W$ into four parts meeting the four demands.  Attach
those parts to the corresponding anchored cores.  The result is a balanced
four-partition of value at least one in every cell.  This completes the final
occupancy pattern and the proof of the lemma.
\end{proof}

\subsection{Envy-free matching and completion}

The matching step is purely graph-theoretic.  A matching between agents and
bundles is \emph{envy-free} if no unmatched agent is adjacent to a matched
bundle.  The following special case is the form needed here; it is part of the
envy-free matching framework of Aigner-Horev and Segal-Halevi
~\cite{aignerhorev2022envyfree}.

\begin{lemma}[Envy-free matching]\label{lem:matching}
Let $A$ and $B$ be finite sets of the same positive size, with a bipartite
graph between them.  If some vertex $d\in A$ is adjacent to every vertex of
$B$, then there is a nonempty matching such that no unmatched vertex of $A$
is adjacent to a matched vertex of $B$.
\end{lemma}

\begin{proof}
For $X\subseteq A$, write $N(X)$ for its neighborhood and define its
deficiency by $\delta(X)=|X|-|N(X)|$.  Choose a set $X$ of maximum deficiency,
and among those choose one of minimum cardinality.

The graph induced between $A\setminus X$ and $B\setminus N(X)$ satisfies
Hall's condition.  Indeed, if some $S\subseteq A\setminus X$ had fewer than
$|S|$ neighbors outside $N(X)$, then
\[
  \delta(X\cup S)
  =\delta(X)+|S|-|N(S)\setminus N(X)|
  >\delta(X),
\]
contradicting maximality.  Hall's theorem therefore gives a matching that
saturates $A\setminus X$ using only bundles in $B\setminus N(X)$.

This matching is nonempty.  If the maximum deficiency is zero, minimality
forces $X=\varnothing$.  If it is positive, $X$ cannot contain $d$, since
then $N(X)=B$ and $\delta(X)=|X|-|B|\leq0$.  Thus in either case
$A\setminus X$ is nonempty.  Finally, the unmatched vertices lie in $X$, and
by construction they have no neighbor among the matched bundles in
$B\setminus N(X)$.
\end{proof}

\begin{proof}[Proof of Theorem~\ref{thm:main}]
The conclusion is immediate when every target is zero, so assume that some
target is positive.  We work first in the ordered normalized instance.  A
positive-target agent $i$ accepts a bundle if its $\bar w_i$-value is at least
one.  A zero-target agent accepts every bundle.

Call a partial allocation \emph{feasible} if:
\begin{enumerate}
  \item its assigned bundles are pairwise disjoint and balanced;
  \item every assigned positive-target agent receives a bundle she accepts;
  \item every unassigned positive-target agent values every assigned bundle
        strictly below one.
\end{enumerate}

The empty partial allocation is feasible.  Among all feasible partial
allocations, choose one assigning as many agents as possible.  Suppose some
positive-target agent $i$ is unassigned.  If $k$ agents have already been
assigned, then $k<4$, and condition~3 says that $i$ values each of their
balanced bundles below one.  Apply Lemma~\ref{lem:residual} to $i$'s
valuation.  Agent $i$ can partition all residual goods into $4-k$ balanced
bundles that she values at least one.

Form the acceptability graph on the $4-k$ unassigned agents and the $4-k$
proposed bundles.  Because $i$ accepts every proposal,
Lemma~\ref{lem:matching} gives a nonempty envy-free matching.

Allocate the matched bundles and retain the old allocation; unmatched proposed
bundles return to the residual pool.  Every newly matched positive agent
receives an acceptable bundle.  Every positive agent left unmatched has no edge to a
newly allocated bundle, so she values it below one; she also continues to
reject all old bundles.  The enlarged partial allocation is therefore
feasible, contradicting maximality.

Hence every positive-target agent is assigned in the maximal partial
allocation.  Give every unallocated good to an arbitrary agent.  This cannot
hurt an assigned agent because valuations are nonnegative, and every
zero-target agent is satisfied even with an empty bundle.  We obtain a
complete ordered allocation $(A_1,A_2,A_3,A_4)$ in which each positive agent
$i$ receives $\bar w_i(A_i)\geq1$.

Transfer the ordered allocation back to the original goods by the picking
sequence described above, obtaining bundles $(S_1,S_2,S_3,S_4)$.  By
\eqref{eq:ordered-transfer}, every positive agent has $w_i(S_i)\geq1$.
Finally apply~\eqref{eq:domination}: she receives original value at least
$t_i=\mms_i^5(M)$.  Every zero-target inequality follows from
nonnegativity.  This proves Theorem~\ref{thm:main}.
\end{proof}

\begin{corollary}[Complete four-agent characterization]
\label{cor:characterization}
Let $d\geq1$ be an integer.  A complete 1-out-of-$d$ MMS allocation is
guaranteed in every four-agent instance with nonnegative additive valuations
if and only if $d\geq5$.
\end{corollary}

\begin{proof}
The maximin share weakly decreases with the denominator.  Indeed, merging any
two cells of a $(d+1)$-partition produces a $d$-partition whose least cell has
no smaller value.  Hence
\[
  \mms_i^d(M)\geq\mms_i^{d+1}(M).
\]
For $d\geq5$, Theorem~\ref{thm:main} therefore gives the desired allocation.
For $d\leq4$, take a four-agent instance with no exact MMS allocation
~\cite{kurokawa2018fair}.  Any allocation meeting every
1-out-of-$d$ share would also meet every 1-out-of-4 share, a contradiction.
\end{proof}

\section{Machine verification}

The theorem is formalized in Lean~4 v4.31.0 against Mathlib v4.31.0.  The
development is sorry-free and uses no project-specific axioms; Lean reports
only \texttt{propext}, \texttt{Classical.choice}, and \texttt{Quot.sound}.
The Lean sources and build instructions will accompany the arXiv submission.
Language models assisted proof search and formalization; Lean checked every
resulting proof term.

\section{Discussion}

Corollary~\ref{cor:characterization} closes the four-agent ordinal-MMS
question: denominator five is sufficient and best possible.  The proof is
existential and does not address computational complexity.

The proof uses three numerical facts that are special to this case: normalized
total value five, four anchors, and the cutoff $1/3$.  Below the cutoff, three
minimal fills leave one unit for the last anchor.  Above it, a unit witness
cell holds at most two anchors, reducing the obstruction to the three
occupancy patterns in Lemma~\ref{lem:residual}.  For more agents, several
collision cells may interact, and the two-error repair from the $2+2$ case no
longer closes the accounting by itself.

The residual lemma, rather than a particular execution of Lone Divider, is the
part of the argument that may extend.  Within this architecture, a general
1-out-of-$(n+1)$ result would require an analogue ensuring that balanced
bundles rejected earlier can be removed without destroying the next divider's
ability to partition the remainder.  For $n=4$, nearest-crossing repair
establishes this property.  An extension along these lines would need a
replacement capable of repairing many simultaneous anchor collisions.

\end{document}